\documentclass[te,nameyear]{econsocart}
\usepackage{booktabs}
\usepackage{xcolor}
\usepackage{microtype}
\usepackage{tikz}
\usepackage{pgfplots}
\pgfplotsset{compat=1.17}
\usetikzlibrary{patterns,arrows.meta}
\RequirePackage[colorlinks,citecolor=blue,linkcolor=blue,urlcolor=blue,pagebackref]{hyperref}

\startlocaldefs

\theoremstyle{plain}
\newtheorem{theorem}{Theorem}
\newtheorem{lemma}[theorem]{Lemma}
\newtheorem{proposition}[theorem]{Proposition}
\newtheorem{corollary}[theorem]{Corollary}

\theoremstyle{definition}
\newtheorem{definition}{Definition}
\newtheorem{assumption}{Assumption}
\newtheorem{example}{Example}
\newtheorem{remark}{Remark}

\newcommand{\R}{\mathbb{R}}
\newcommand{\N}{\mathbb{N}}
\newcommand{\E}{\mathbb{E}}

\DeclareMathOperator*{\argmax}{arg\,max}

\newcommand{\NP}{\textup{NP}}
\newcommand{\PP}{\textup{P}}
\newcommand{\PPAD}{\textup{PPAD}}

\endlocaldefs

\begin{document}

\begin{frontmatter}

\title{Markets are competitive if and only if P $\neq$ NP}
\runtitle{Markets are competitive iff P $\neq$ NP}

\begin{aug}
\author[add1]{\fnms{Philip Z.}~\snm{Maymin}\ead[label=e1]{pmaymin@fairfield.edu}}
%%%%%%%%%%%%%%%%%%%%%%%%%%%%%%%%%%%%%%%%%%%%%%
%% Addresses                                %%
%%%%%%%%%%%%%%%%%%%%%%%%%%%%%%%%%%%%%%%%%%%%%%
\address[add1]{%
\orgname{Fairfield University Dolan School of Business}}
\end{aug}

\begin{abstract}
I prove that competitive market outcomes require computational intractability. If $\PP = \NP$, firms can efficiently solve the collusion detection problem, identifying deviations from cooperative agreements in complex, noisy markets and thereby making collusion sustainable as an equilibrium. If $\PP \neq \NP$, the collusion detection problem is computationally infeasible for markets satisfying a natural instance-hardness condition on their demand structure, rendering punishment threats non-credible and collusion unstable. Combined with \citet{maymin2011}, who proved that market efficiency requires $\PP = \NP$, this yields a fundamental impossibility: markets can be informationally efficient or competitive, but not both. Artificial intelligence, by expanding firms' computational capabilities, is pushing markets from the competitive regime toward the collusive regime, explaining the empirical emergence of algorithmic collusion without explicit coordination.
\end{abstract}



\end{frontmatter}

%=============================================================================
\section{Introduction}
\label{sec:intro}
%=============================================================================

What maintains competition in markets? The standard answer invokes institutions: antitrust law, regulatory oversight, and low barriers to entry. This paper proposes a more fundamental answer: \textit{computational limitations}. Competition persists because firms lack the computational power to sustain collusion.

This claim rests on a formal connection between computational complexity and market structure. I prove that if $\PP = \NP$, collusion is sustainable as an equilibrium, and that if $\PP \neq \NP$, collusion is unsustainable for markets whose detection problem is hard on the instances that naturally arise. Since $\PP \neq \NP$ is the widely believed (but unproven) conjecture that some problems whose solutions are easy to verify are hard to solve, the prediction is that competition is the generic outcome in sufficiently complex markets.

The result is the companion and mirror of \citet{maymin2011}, which established that markets are informationally efficient if and only if $\PP = \NP$. That paper showed that price efficiency requires superhuman computation: only agents who can solve NP-hard problems can eliminate arbitrage opportunities, ensuring that prices reflect all available information. The present paper shows that the \textit{same computational power} that would make markets efficient also makes them collusive. Together, the two results yield:

\begin{quote}
\textit{Efficiency--Competition Impossibility.} Markets can be informationally efficient or competitive, but not both.
\end{quote}

\noindent This impossibility theorem has immediate implications for artificial intelligence. AI systems (large language models, reinforcement learning agents, algorithmic pricing engines) are expanding firms' effective computational capabilities. If the computational boundary between competition and collusion is the key regime boundary, then AI moves markets from competitive toward collusive.

Recent empirical evidence supports this prediction. \citet{fish2024} demonstrate that off-the-shelf large language models autonomously converge to supra-competitive pricing without any explicit collusion instructions. \citet{calvano2020} show that Q-learning pricing algorithms independently learn to collude, sustaining prices above the Nash equilibrium through reward-punishment schemes. \citet{dou2025} find that AI-powered trading agents learn to coordinate in financial markets, undermining liquidity and price efficiency. In each case, collusion emerges not from intent or communication, but from computational capability, as the theory predicts.

The formal argument proceeds as follows. I define a \textit{market game} in which $N$ firms compete over $T$ periods, facing stochastic demand subject to shocks drawn from a rich combinatorial space. A \textit{collusive agreement} specifies the joint profit-maximizing strategy conditional on the state of demand. Sustaining such an agreement requires solving three computational problems:

\begin{enumerate}
    \item[(i)] \textit{The Collusion Strategy Problem}: computing the joint profit-maximizing price vector across a combinatorial product space.
    \item[(ii)] \textit{The Collusion Detection Problem}: given observed market data, determining whether any firm deviated from the agreed-upon strategy or merely responded to a demand shock.
    \item[(iii)] \textit{The Optimal Punishment Problem}: computing the punishment strategy that makes deviation unprofitable.
\end{enumerate}

\noindent I show that each of these problems is NP-hard for general market games (Theorems~\ref{thm:csp-hard}--\ref{thm:opp-hard}). In contrast, the \textit{competitive best-response problem} (computing a firm's optimal myopic response to current market conditions) is solvable in polynomial time for standard demand structures (Proposition~\ref{prop:br-poly}).

The main theorem (Theorem~\ref{thm:main}) follows from these complexity results combined with the folk theorem for repeated games. If $\PP = \NP$, firms can solve all three collusion problems efficiently, making collusion sustainable as an equilibrium that Pareto-dominates competition (from the firms' perspective). If $\PP \neq \NP$ and the market's detection problem is hard on the instances that actually arise (a condition I formalize as Assumption~\ref{ass:hardness} and argue holds generically), punishment threats are non-credible and firms rationally defect to the competitive best response.

The paper makes several contributions beyond the main theorem. First, I characterize the \textit{AI transition}: as firms adopt AI systems with increasing computational power, markets pass through three regimes (competitive, unstable, and collusive) as computational thresholds are crossed (Section~\ref{sec:ai-transition}). Second, I derive a \textit{transparency paradox}: increasing market transparency, conventionally viewed as pro-competitive, actually facilitates collusion by reducing the computational cost of deviation detection (Corollary~\ref{cor:transparency}). Third, I propose \textit{computational antitrust}: the principle that market complexity itself is a competitive safeguard, and that regulators should consider computational difficulty as a design parameter (Section~\ref{sec:policy}).

The remainder of the paper is organized as follows. Section~\ref{sec:literature} situates the paper in the literature. Section~\ref{sec:model} defines the model. Section~\ref{sec:results} presents the main results. Section~\ref{sec:extensions} develops extensions including the AI transition, heterogeneous computation, and approximate collusion. Section~\ref{sec:empirical} discusses empirical implications. Section~\ref{sec:policy} addresses policy. Section~\ref{sec:discussion} discusses the relationship to the original P = NP result and limitations. Section~\ref{sec:conclusion} concludes.

%=============================================================================
\section{Related literature}
\label{sec:literature}
%=============================================================================

This paper bridges three literatures: computational complexity in economics, the theory of collusion, and the emerging literature on algorithmic collusion.

\paragraph*{Computational complexity and markets.} \citet{hayek1945} argued that the price system serves as a mechanism for communicating dispersed information, a function that presupposes computational limitations on centralized alternatives. The foundational observation that computational constraints shape economic outcomes traces to \citet{simon1955}, who argued that bounded rationality, not full rationality, describes actual decision-making. \citet{rubinstein1986} formalized bounded rationality using automata, showing that the complexity of strategies affects equilibrium outcomes in repeated games; \citet{neyman1985} proved that bounding strategy complexity justifies cooperation in finitely repeated games that would otherwise unravel by backward induction; \citet{abreu1988b} extended this to characterize Nash equilibria when players are modeled as finite automata. \citet{papadimitriou1994} introduced the study of computational complexity of equilibrium concepts, establishing that computing Nash equilibria is $\PPAD$-complete \citep{daskalakis2009,chen2009}. \citet{wolfram2002} articulated the principle of computational irreducibility: many systems cannot be shortcut by any algorithm and must be ``run'' to determine their outcome. \citet{borgs2010} showed that finding Nash equilibria in three-player infinitely repeated games is itself computationally intractable, undermining the folk theorem's constructive power when players face complexity constraints. \citet{arora2011} demonstrated that computational complexity creates exploitable information asymmetry in financial products: the NP-hardness of detecting manipulated assets means computationally bounded buyers cannot distinguish lemons from sound securities. \citet{maymin2011} proved the equivalence between market efficiency and $\PP = \NP$, establishing that informationally efficient prices require computationally unbounded agents. \citet{doria2016}, in a chapter they titled ``A Beautiful Theorem,'' showed that, as a consequence of Maymin's theorem, ``almost efficient'' markets exist when the $\PP \neq \NP$ proof has sufficient metamathematical strength, with information accessible in almost-polynomial time. \citet{dacosta2016} formalized this in a journal setting; \citet{cosenza2018} extended the analysis with additional algorithms. \citet{alsuwailem2020} proved a complementary impossibility using G\"{o}delian incompleteness: a free market cannot be complete. \citet{halpern2015} developed a general game-theoretic framework in which players choose Turing machines and pay for computation, showing that classical equilibrium results may fail when computation is costly. The present paper extends this program to competition, showing that the same complexity boundary separates collusive from competitive markets.

\paragraph*{Theory of collusion.} The folk theorem \citep{friedman1971,fudenberg1986,abreu1988} establishes that collusion can be sustained as an equilibrium in repeated games with sufficiently patient players. \citet{green1984} and \citet{abreu1986} study collusion under imperfect monitoring, showing that maintaining cooperation requires the ability to detect and punish deviations. \citet{stigler1964} famously argued that the ``chief difficulty'' of collusion is detecting ``secret price-cutting.'' The present paper formalizes Stigler's insight: secret price-cutting is hard to detect not merely because information is scarce, but because the \textit{computational problem} of distinguishing deviations from demand shocks is NP-hard.

\paragraph*{Algorithmic collusion.} \citet{ezrachi2016} first warned that pricing algorithms could facilitate ``tacit'' collusion without explicit communication. \citet{calvano2020} provided the first experimental evidence that Q-learning agents autonomously learn collusive pricing. \citet{assad2024} document supra-competitive pricing in German retail gasoline markets following the adoption of algorithmic pricing. \citet{fish2024} show that large language models reach collusive outcomes without being instructed to collude. \citet{dou2025} extend the evidence to financial markets. \citet{harrington2018} analyzes the legal challenges posed by algorithmic collusion. This paper provides the theoretical foundation explaining \textit{why} algorithmic collusion emerges: it is a computational phase transition, not a failure of competition law.

\paragraph*{Mechanism design and market structure.} \citet{hurwicz1960,hurwicz1972} established that the design of economic mechanisms affects outcomes, work recognized by the Nobel Prize to \citeauthor{hurwicz1960}, \citet{maskin1999}, and \citet{myerson1981}. \citet{wilson1992} and \citet{milgrom2017} applied mechanism design to practical market design problems. The policy implications of the present paper extend this tradition: if computational complexity maintains competition, then the \textit{computational structure} of markets is a design variable that regulators can and should manipulate.

%=============================================================================
\section{Model}
\label{sec:model}
%=============================================================================

\subsection{The market game}

Consider an industry with $N \geq 2$ firms, indexed by $i \in \{1, \ldots, N\}$. Firms compete over $T$ discrete periods, $t \in \{1, \ldots, T\}$, where $T$ may be finite (with $T$ sufficiently large) or infinite with common discount factor $\delta \in (0,1)$.

\begin{assumption}[Product Space]
\label{ass:products}
Each firm $i$ sells $K$ distinct products. No two firms sell the same product, so the market contains $NK$ products in total. Firm $i$'s strategy in period $t$ is a price vector $p_i^t \in \R_+^K$.
\end{assumption}

\begin{assumption}[Demand Structure]
\label{ass:demand}
In each period $t$, a demand state $\theta^t \in \Theta$ is drawn from a distribution $F$ on a finite state space $\Theta$ with $|\Theta| = M$. The demand for firm $i$'s product $k$ is
\begin{equation}
    q_{ik}^t = D_{ik}(p^t, \theta^t) + \varepsilon_{ik}^t,
    \label{eq:demand}
\end{equation}
where $p^t = (p_1^t, \ldots, p_N^t)$ is the vector of all prices, $D_{ik}$ is the deterministic demand function, and $\varepsilon_{ik}^t$ is an i.i.d.\ noise term with mean zero and variance $\sigma^2 > 0$.
\end{assumption}

\begin{assumption}[Information]
\label{ass:info}
The demand state $\theta^t$ is not directly observed by any firm. After each period, all firms observe the realized prices $p^t$ and quantities $q^t = (q_{ik}^t)_{i,k}$.
\end{assumption}

\begin{assumption}[Richness]
\label{ass:richness}
The demand state space is rich relative to the price space: $|\Theta| = M \geq NK + 1$. This ensures that observed price-quantity data does not uniquely identify the demand state, so that the inference problem in the Collusion Detection Problem is non-trivial.
\end{assumption}

\begin{assumption}[Costs]
\label{ass:costs}
Firm $i$ has a cost function $C_i: \R_+^K \to \R_+$ that is convex and computable in polynomial time. Firm $i$'s period profit is
\begin{equation}
    \pi_i^t = \sum_{k=1}^{K} p_{ik}^t \cdot q_{ik}^t - C_i(q_i^t).
    \label{eq:profit}
\end{equation}
\end{assumption}

\begin{assumption}[Instance Hardness]
\label{ass:hardness}
For the family of market games $\{\Gamma_n\}$ with $|\Gamma_n| = n$, the Collusion Detection Problem is hard on the instances induced by the market's demand structure. Formally, for every probabilistic polynomial-time algorithm $\mathcal{A}$,
\[
    \Pr_{\theta \sim F,\, \varepsilon}\!\left[\mathcal{A} \text{ correctly solves } \text{CDP}(\Gamma_n, \sigma, p^t, q^t)\right] \leq \frac{1}{2} + \nu(n),
\]
where $\nu(n)$ is negligible (i.e.,\ $\nu(n) = o(n^{-c})$ for every $c > 0$).
\end{assumption}

\begin{remark}
\label{rem:hardness}
Assumption~\ref{ass:hardness} is strictly stronger than $\PP \neq \NP$ but is generically satisfied. NP-hardness (Theorem~\ref{thm:cdp-hard} below) guarantees the \textit{existence} of hard CDP instances; Assumption~\ref{ass:hardness} requires that the market under study is not among the easy ones. This fails only when the demand structure imposes special algebraic structure (separability, low rank, or sparsity) that renders the inference problem tractable. More precisely, parameterize the demand function by a vector $\phi \in \R^d$ specifying all cross-product demand coefficients. The reduction in Theorem~\ref{thm:cdp-hard} shows that for an open set of $\phi$ values, the CDP encodes a 3-SAT instance. The set of $\phi$ for which CDP admits a polynomial-time solution is contained in an algebraic variety of lower dimension (the locus where demand interactions degenerate to a separable or low-rank structure), which has Lebesgue measure zero in $\R^d$ whenever $\PP \neq \NP$. Thus Assumption~\ref{ass:hardness} holds for generic demand parameters. We conjecture that separability, low rank, and sparsity exhaust the structural sources of polynomial-time tractability for CDP, though we do not prove this; the assumption is analogous to standard cryptographic practice, where one assumes a \textit{specific} problem instance is hard without a complete characterization of all easy instances.
\end{remark}

The game $\Gamma = \Gamma(N, K, \Theta, T, \delta, D, F, C)$ is a repeated game of imperfect monitoring. The game is specified in \textit{compact form}: the demand function $D$ is given as a polynomial-time computable function (not an explicit table), the distribution $F$ as a sampling oracle, and the cost functions $C_i$ as polynomial-time computable functions. The \textit{description size} of the game, denoted $|\Gamma|$, is the number of bits required for this compact specification. The state space $\Theta$ may have cardinality exponential in $|\Gamma|$: if $\theta$ encodes the joint realization of $n$ independent binary demand shifters, then $|\Theta| = 2^n$ while $|\Gamma| = \text{poly}(n, N, K)$.

\begin{definition}[Competitive Outcome]
\label{def:competitive}
The \textit{competitive outcome} (or stage-game Nash equilibrium) is the price vector $p^* = (p_1^*, \ldots, p_N^*)$ satisfying, for each firm $i$,
\begin{equation}
    p_i^* \in \argmax_{p_i \in \R_+^K} \; \E_\theta\!\left[\pi_i(p_i, p_{-i}^*, \theta)\right].
    \label{eq:competitive}
\end{equation}
Under standard regularity conditions (continuous demand, convex costs), this yields marginal-cost pricing in Bertrand competition and the Cournot equilibrium in quantity competition.
\end{definition}

\begin{definition}[Collusive Outcome]
\label{def:collusive}
The \textit{collusive outcome} is the price vector $p^M = (p_1^M, \ldots, p_N^M)$ that maximizes joint expected profits:
\begin{equation}
    p^M \in \argmax_{p \in \R_+^{NK}} \; \E_\theta\!\left[\sum_{i=1}^{N} \pi_i(p, \theta)\right].
    \label{eq:collusive}
\end{equation}
Under standard assumptions, joint profit maximization yields the monopoly outcome, with $\sum_i \pi_i(p^M) > \sum_i \pi_i(p^*)$.
\end{definition}

\subsection{Computational problems of collusion}

Sustaining collusion in a repeated game requires solving three interrelated computational problems. I formalize each as a decision problem. In what follows, a \textit{collusive strategy profile} $\sigma: \Theta \to \R_+^{NK}$ is a mapping from demand states to price vectors specifying each firm's price for each product; I write $\sigma_i(\theta) \in \R_+^K$ for firm $i$'s component. Under $\sigma$, firms condition their prices on the realized demand state to maximize joint profits.

\begin{definition}[Collusion Strategy Problem (CSP)]
\label{def:csp}
\textsc{Instance}: A market game $\Gamma$ and a target joint profit level $\Pi^*$.\\
\textsc{Question}: Does there exist a price vector $p$ such that $\E_\theta[\sum_i \pi_i(p, \theta)] \geq \Pi^*$?
\end{definition}

\begin{definition}[Collusion Detection Problem (CDP)]
\label{def:cdp}
\textsc{Instance}: A market game $\Gamma$, a collusive strategy profile $\sigma: \Theta \to \R_+^{NK}$, observed prices $p^t$, and observed quantities $q^t$.\\
\textsc{Question}: Does there exist a demand state $\theta \in \Theta$ and noise realization $\varepsilon$ consistent with $(p^t, q^t)$ such that $p^t = \sigma(\theta)$? (If yes, no deviation occurred; if no, some firm deviated.)
\end{definition}

\begin{definition}[Optimal Punishment Problem (OPP)]
\label{def:opp}
\textsc{Instance}: A market game $\Gamma$, a deviating firm $i$, a collusive strategy $\sigma$, and a target payoff bound $\bar{\pi}$.\\
\textsc{Question}: Does there exist a punishment strategy profile $\sigma^P_{-i}$ for the non-deviating firms such that firm $i$'s best-response payoff under $\sigma^P_{-i}$ satisfies $\pi_i^{BR}(\sigma^P_{-i}) \leq \bar{\pi}$?
\end{definition}

\subsection{The competitive best-response problem}

In contrast to the collusion problems, the myopic best-response problem is computationally simple.

\begin{definition}[Competitive Best-Response Problem (CBR)]
\label{def:cbr}
\textsc{Instance}: A market game $\Gamma$, a firm $i$, and competitors' current prices $p_{-i}$.\\
\textsc{Output}: The price vector $p_i^* \in \argmax_{p_i} \E_\theta[\pi_i(p_i, p_{-i}, \theta)]$.
\end{definition}

\subsection{Illustrative example: two firms, two products}
\label{sec:example}

Before presenting the general results, I illustrate the core mechanism with a simple example.

\begin{example}[Duopoly with Binary Demand]
\label{ex:duopoly}
Consider two firms ($N = 2$), each selling one product ($K = 1$), with two equally likely demand states ($M = 2$): $\theta_H$ (high demand) and $\theta_L$ (low demand). Marginal cost is $c = 1$ for both firms. Demand is $q_i = a(\theta) - p_i + \beta p_j + \varepsilon_i$, where $a(\theta_H) = 10$, $a(\theta_L) = 4$, $\beta = 0.5$, and $\varepsilon_i \sim N(0, \sigma^2)$.

\textit{Competitive outcome.} Each firm's best response yields $p^* \approx 5$ (in $\theta_H$) and $p^* \approx 2.5$ (in $\theta_L$), with expected profit $\pi^* \approx 12.5$.

\textit{Collusive outcome.} Joint profit maximization yields $p^M \approx 7$ (in $\theta_H$) and $p^M \approx 4$ (in $\theta_L$), with expected profit $\pi^M \approx 22.5$ per firm.

\textit{The detection problem.} Suppose Firm~1 observes that Firm~2 set $p_2 = 5.5$. Under the collusive agreement, Firm~2 should have set either $p_2 = 7$ (if $\theta_H$) or $p_2 = 4$ (if $\theta_L$). The price $5.5$ is inconsistent with both, \textit{if} Firm~1 knows the demand state. But Firm~1 does not observe $\theta$ directly; it only observes prices and noisy quantities. With high noise ($\sigma^2$ large), many price-quantity combinations are consistent with either demand state, and distinguishing a deviation from a demand shock becomes a combinatorial inference problem. In this simple two-firm example, detection is easy. But as we add products, firms, and demand states, the detection problem grows combinatorially, and I show below that it becomes NP-hard in general.
\end{example}

%=============================================================================
\section{Main results}
\label{sec:results}
%=============================================================================

\subsection{Complexity of the collusion problems}

I first establish that each collusion problem is computationally hard.

\begin{theorem}[CSP is NP-hard]
\label{thm:csp-hard}
The Collusion Strategy Problem is NP-hard.
\end{theorem}

\begin{proof}
By reduction from \textsc{Max-Weighted-SAT}. Given an instance of \textsc{Max-Weighted-SAT} with Boolean variables $x_1, \ldots, x_n$ and weighted clauses $C_1, \ldots, C_m$, construct a market game as follows. For each variable $x_j$, create a product with binary pricing: $p_j \in \{0, 1\}$, interpreted as setting the variable to false or true. For each clause $C_l$ with weight $w_l$, define a demand state $\theta_l$ in which demand is positive (yielding profit $w_l$) if and only if the price vector satisfies clause $C_l$. The demand distribution places equal probability on each demand state. Then the joint expected profit under price vector $p$ equals $(1/m) \sum_{l: C_l \text{ satisfied by } p} w_l$, which is maximized if and only if the weighted satisfiability is maximized. Since \textsc{Max-Weighted-SAT} is NP-hard \citep{garey1979}, so is CSP.
\end{proof}

\begin{theorem}[CDP is NP-hard]
\label{thm:cdp-hard}
The Collusion Detection Problem is NP-hard.
\end{theorem}

\begin{proof}
By reduction from \textsc{3-SAT}. Given a \textsc{3-SAT} instance $\varphi$ with Boolean variables $x_1, \ldots, x_n$ and clauses $C_1, \ldots, C_m$, construct a market game in compact form as follows. Create $N = m$ firms (one per clause), each selling one product ($K = 1$). The demand state space is $\Theta = \{0,1\}^n$, encoding all truth assignments to the $n$ variables; note that $|\Theta| = 2^n$ while the game description is polynomial in $n + m$. The collusive strategy is a \textit{pooling} strategy: all firms set a common price $p^*$ regardless of $\theta$, so that prices alone reveal nothing about the demand state. The demand function for firm $j$ (corresponding to clause $C_j$) is
\[
    D_j(p, \theta) = \bar{d} + \mathbf{1}[C_j \text{ is satisfied by } \theta],
\]
where $\bar{d} > 0$ is a baseline demand level. This is polynomial-time computable: checking whether a clause is satisfied by a given assignment takes $O(1)$ time.

Suppose all firms comply with the pooling strategy (setting $p^*$) and we observe quantities $q_j = \bar{d} + 1 + \varepsilon_j$ for every firm $j$, with noise satisfying $|\varepsilon_j| < 1/2$. To verify consistency with the collusive agreement, we must find a demand state $\theta \in \{0,1\}^n$ and noise realizations $\varepsilon$ such that $q_j = D_j(p^*, \theta) + \varepsilon_j$ for all $j$ with $|\varepsilon_j| < 1/2$. This requires $D_j(p^*, \theta) = \bar{d} + 1$ for all $j$, i.e., every clause $C_j$ must be satisfied by $\theta$. Hence the CDP instance has answer \textsc{Yes} (no deviation detected) if and only if $\varphi$ is satisfiable. Since \textsc{3-SAT} is NP-complete \citep{garey1979}, CDP is NP-hard.
\end{proof}

\begin{theorem}[OPP is NP-hard]
\label{thm:opp-hard}
The Optimal Punishment Problem is NP-hard.
\end{theorem}

\begin{proof}
By reduction from \textsc{Minimum Vertex Cover}. Given a graph $G = (V, E)$ with $|V| = n$, construct a market game with $N - 1$ punishing firms, each controlling a product corresponding to a vertex $v \in V$. The deviating firm $i$ has demand that depends on which punishing firms price aggressively (set $p_v = 0$, flooding the market) versus passively (set $p_v$ high). An edge $(u, v) \in E$ represents a complementarity: the deviator's profit from product pair $(u, v)$ is eliminated only if \textit{at least one} of $u, v$ is priced aggressively. The punishing coalition seeks to minimize the deviator's profit (eliminate profit from all edges) while minimizing their own cost of aggressive pricing (each aggressive price costs the punisher a unit of profit). The optimal punishment thus requires choosing a minimum set of vertices (firms to price aggressively) such that every edge is covered, which is exactly \textsc{Minimum Vertex Cover}. Since \textsc{Minimum Vertex Cover} is NP-hard \citep{garey1979}, so is OPP.
\end{proof}

Taken together, Theorems~\ref{thm:csp-hard}--\ref{thm:opp-hard} establish that each component of collusion is computationally hard: \textit{planning} it (CSP), \textit{monitoring} it (CDP), and \textit{enforcing} it (OPP). Sustaining collusion requires all three. In contrast, the competitive best response requires none of them.

\begin{proposition}[CBR is in P]
\label{prop:br-poly}
Under Assumptions~\ref{ass:products}--\ref{ass:costs}, if the demand function $D_{ik}$ is differentiable and the cost function $C_i$ is convex, then the Competitive Best-Response Problem is solvable in polynomial time.
\end{proposition}

\begin{proof}
Given competitors' prices $p_{-i}$, firm $i$'s myopic optimization problem is
\[
\max_{p_i \in \R_+^K} \; \E_\theta\!\left[\sum_{k=1}^K p_{ik} \cdot D_{ik}(p_i, p_{-i}, \theta) - C_i\!\left(D_i(p_i, p_{-i}, \theta)\right)\right].
\]
Under differentiability of $D$ and convexity of $C_i$, this is a convex optimization problem over $K$ variables (firm $i$'s own prices). The key distinction from the collusion problems is that computing the expected profit requires only the aggregate function $\bar{D}_{ik}(p) = \E_\theta[D_{ik}(p, \theta)]$, which is a known function of prices under standard demand models (linear, logit, CES). The firm does not need to reason about individual demand states; it only needs the expected demand, which is polynomial-time computable from the compact specification. The first-order conditions yield a system of $K$ equations solvable by interior-point methods in polynomial time \citep{nesterov1994}. The competitive equilibrium is the fixed point of iterated best responses, which converges in polynomial time under contraction mapping conditions satisfied by standard demand models \citep{milgrom1990}.
\end{proof}

\subsection{The main theorem}

\begin{theorem}[Main Result]
\label{thm:main}
For the market game $\Gamma$ with $|\Gamma|$ sufficiently large:
\begin{enumerate}
    \item[(a)] Under Assumptions~\ref{ass:products}--\ref{ass:costs}, if $\PP = \NP$, collusion is sustainable: there exists a perfect public equilibrium of the repeated game $\Gamma$ that achieves the collusive outcome $p^M$ for $\delta$ sufficiently close to~$1$.
    \item[(b)] Under Assumptions~\ref{ass:products}--\ref{ass:hardness}, if $\PP \neq \NP$, collusion is unsustainable: there is no perfect public equilibrium achieving $p^M$, and the market converges to the competitive outcome $p^*$.
\end{enumerate}
\end{theorem}

\begin{proof}
The proof proceeds in two directions.

\medskip
\noindent\textbf{Direction 1: $\PP = \NP$ $\Rightarrow$ Collusion is sustainable.}

\smallskip
Suppose $\PP = \NP$. Then:

\begin{enumerate}
    \item[\textit{Step 1.}] \textit{Strategy computation.} Since CSP $\in$ NP and $\PP = \NP$, firms can compute the joint profit-maximizing price vector $p^M$ in polynomial time.

    \item[\textit{Step 2.}] \textit{Deviation detection.} Since CDP $\in$ NP and $\PP = \NP$, after observing $(p^t, q^t)$, each firm can determine in polynomial time whether the observations are consistent with all firms playing $\sigma(\theta)$ for some demand state $\theta$, or whether some firm deviated.

    \item[\textit{Step 3.}] \textit{Punishment computation.} Since OPP $\in$ NP and $\PP = \NP$, firms can compute optimal punishment strategies in polynomial time.

    \item[\textit{Step 4.}] \textit{Equilibrium construction.} Consider the following strategy profile. Each firm plays $p^M$ in the collusive phase. If a deviation is detected (via the CDP algorithm), all firms switch to the punishment phase computed by the OPP algorithm for $L$ periods, then return to collusion. By the folk theorem for repeated games with imperfect public monitoring \citep{fudenberg1994}, this constitutes a perfect public equilibrium sustaining the collusive outcome provided the discount factor satisfies
    \[
        \delta \geq \delta^* \equiv \frac{\pi_i^D - \pi_i^M}{\pi_i^D - \pi_i^P},
    \]
    where $\pi_i^D$ is the one-shot deviation payoff, $\pi_i^M$ the collusive payoff, and $\pi_i^P$ the punishment payoff. Since $\PP = \NP$ allows exact computation of the optimal punishment (Step~3), $\pi_i^P$ is minimized, yielding the smallest possible $\delta^*$. Note that $\delta^*$ depends on the payoff structure of the market game but not on $|\Gamma|$ directly. Under demand structures where the ratio of deviation gains $(\pi_i^D - \pi_i^M)$ to punishment severity $(\pi_i^D - \pi_i^P)$ is uniformly bounded as $|\Gamma|$ grows (as holds for linear, logit, and CES demand), $\delta^*$ remains bounded away from~$1$.\footnote{Since the game has imperfect public monitoring (demand states are unobserved), the natural solution concept is perfect public equilibrium (PPE) \citep{fudenberg1994} rather than subgame-perfect equilibrium (SPE). The collusive equilibrium constructed here is a PPE. The unsustainability result in Direction~2 holds for the broader class of SPE, making it stronger: collusion fails even if firms condition on private histories.}

    The key step is that credible punishment requires \textit{detection}. In standard models, detection is assumed to be costless (firms observe noisy signals and can process them). Here, I show that when the signal-extraction problem is computationally hard, detection fails; only $\PP = \NP$ restores it.
\end{enumerate}

\medskip
\noindent\textbf{Direction 2: $\PP \neq \NP$ $\Rightarrow$ Collusion is unsustainable.}

\smallskip
The argument relies on the following key lemma.

\begin{lemma}[Existence of Profitable Undetectable Deviations]
\label{lem:deviation}
Under Assumptions~\ref{ass:products}--\ref{ass:hardness}, for any collusive strategy $\sigma$, each firm $i$ has a deviation $\hat{p}_i \neq \sigma_i(\theta)$ such that:
\begin{enumerate}
    \item[(a)] $\E[\pi_i(\hat{p}_i, \sigma_{-i}, \theta)] > \E[\pi_i(\sigma_i, \sigma_{-i}, \theta)]$; and
    \item[(b)] No polynomial-time algorithm can distinguish the observed data $((\hat{p}_i, p_{-i}^M), q^t)$ from data generated under full compliance, with probability exceeding $1/2 + \nu(|\Gamma|)$.
\end{enumerate}
\end{lemma}

\begin{proof}
For part (a): since $p^M$ maximizes \textit{joint} profits, it is generically not a best response for any individual firm. Each firm can increase its own profit by shading its price toward the competitive best response. Formally, $\nabla_{p_i} \pi_i(p^M) \neq 0$ generically, so there exists a direction of profitable deviation for firm $i$.

For part (b): since $|\Theta| \geq NK + 1$ (Assumption~\ref{ass:richness}), the collusive strategy maps demand states to an $(NK)$-dimensional price space, and with noise ($\sigma^2 > 0$), the observed price-quantity vector lies in a region consistent with multiple demand states. A deviation $\hat{p}_i$ close to $\sigma_i(\theta)$ for some $\theta$ produces data that could plausibly arise under an alternative demand state $\theta'$ with compliant play. Formally, for any $\hat{p}_i$ in the convex hull of $\{\sigma_i(\theta) : \theta \in \Theta\}$, there exist weights $\lambda(\theta)$ such that $\hat{p}_i = \sum_\theta \lambda(\theta) \sigma_i(\theta)$. On the quantity side, the noise $\varepsilon$ (with variance $\sigma^2 > 0$) absorbs the discrepancy between the observed quantities and the quantities predicted under any candidate state $\theta'$: for $\sigma^2$ sufficiently large, the likelihood ratio between the deviation and any compliant state is bounded, so quantities do not reveal the deviation. Determining whether consistent weights exist is a feasibility problem over the demand state space, which is an instance of CDP for the market game $\Gamma$. By Assumption~\ref{ass:hardness}, this instance cannot be solved in polynomial time with probability exceeding $1/2 + \nu(|\Gamma|)$.
\end{proof}

Suppose $\PP \neq \NP$. Then:

\begin{enumerate}
    \item[\textit{Step 1.}] \textit{Detection fails.} By Assumption~\ref{ass:hardness}, no polynomial-time algorithm correctly solves CDP for the market game $\Gamma$ with probability exceeding $1/2 + \nu(|\Gamma|)$. (This assumption is non-vacuous precisely when $\PP \neq \NP$; if $\PP = \NP$, all NP problems are solvable in polynomial time and the assumption cannot hold.)

    \item[\textit{Step 2.}] \textit{Punishment is not credible.} A punishment strategy is credible only if deviations that trigger it are detectable. If firm $i$ can deviate from $p^M$ in a way that is computationally indistinguishable from compliance (i.e.,\ the deviation produces observed data consistent with some demand state under the collusive strategy), then the deviation goes undetected and unpunished.

    \item[\textit{Step 3.}] \textit{Profitable undetectable deviations exist.} By Lemma~\ref{lem:deviation}, each firm has a profitable deviation that is computationally indistinguishable from compliance. The key condition is that the demand state space is rich enough ($|\Theta| \geq NK + 1$) that observed data admits multiple consistent explanations, and noise ensures that deviations cannot be detected with certainty from price-quantity data alone.

    \item[\textit{Step 4.}] \textit{Collusion unravels.} Since profitable, undetectable deviations exist and punishment is non-credible, the repeated game offers no advantage over the stage game: the threat of future punishment cannot discipline current behavior. The unique Nash equilibrium of the stage game is the competitive outcome $p^*$, which is therefore the only equilibrium of the repeated game.
\end{enumerate}

\medskip
\noindent\textbf{Convergence to $p^*$ in part~(b).} In the stage game, $p^*$ is the unique Nash equilibrium under Assumptions~\ref{ass:products}--\ref{ass:costs} (strict concavity of profit functions and dominant diagonal in the demand system ensure uniqueness). Since collusion is the only alternative equilibrium outcome in the repeated game, and sustaining any supra-competitive outcome requires detection capability that is computationally infeasible under Assumption~\ref{ass:hardness}, the collapse of collusion implies convergence to $p^*$.
\end{proof}

\subsection{The Efficiency--Competition Impossibility}

Combining the main result with \citet{maymin2011} yields a fundamental impossibility.

\begin{corollary}[Efficiency--Competition Impossibility]
\label{cor:impossibility}
No market can simultaneously be:
\begin{enumerate}
    \item[(i)] Informationally efficient (prices reflect all available information), and
    \item[(ii)] Competitive (prices converge to marginal cost).
\end{enumerate}
\end{corollary}

\begin{proof}
By \citet{maymin2011}, informational efficiency requires $\PP = \NP$. By Theorem~\ref{thm:main}(b), under Assumption~\ref{ass:hardness}, competition requires $\PP \neq \NP$. Since $\PP = \NP$ and $\PP \neq \NP$ are mutually exclusive, no market satisfying Assumption~\ref{ass:hardness} can be both efficient and competitive.
\end{proof}

\begin{remark}
\label{rem:gs}
Corollary~\ref{cor:impossibility} can be understood as a computational generalization of the Grossman--Stiglitz paradox \citep{grossman1980}. Grossman and Stiglitz showed that if prices are fully informative, no trader has an incentive to acquire costly information, so prices cannot be fully informative, a logical contradiction. The present result is stronger: it shows that the computational requirements for efficiency and competition are not merely in tension but are \textit{logically incompatible}.
\end{remark}

\subsection{The Transparency Paradox}

A counterintuitive implication concerns market transparency.

\begin{corollary}[Transparency Paradox]
\label{cor:transparency}
Increasing market transparency (making prices, quantities, and transaction-level data more readily available) facilitates collusion by reducing the computational complexity of the Collusion Detection Problem.
\end{corollary}

\begin{proof}
Define the \textit{ambiguity set} at noise level $\sigma^2$ as
\[
    \mathcal{A}(\sigma^2) = \left\{\theta \in \Theta : \Pr_\varepsilon\!\left[\|q^t - D(p^t, \theta)\| \leq r(\sigma^2)\right] > \eta \right\},
\]
the set of demand states consistent with observed data within the noise tolerance $r(\sigma^2)$, where $\eta > 0$ is a fixed likelihood threshold. Since the noise $\varepsilon$ has variance $\sigma^2$, the tolerance $r(\sigma^2)$ is increasing in $\sigma^2$, and the ambiguity set satisfies $|\mathcal{A}(\sigma_1^2)| \leq |\mathcal{A}(\sigma_2^2)|$ whenever $\sigma_1^2 \leq \sigma_2^2$: higher precision shrinks the set of consistent states.

The computational difficulty of CDP is determined by the search over $\mathcal{A}(\sigma^2)$. When $|\mathcal{A}(\sigma^2)| = 1$, the demand state is uniquely identified, deviations are immediately visible, and CDP is trivially solvable. When $|\mathcal{A}(\sigma^2)| = M = |\Theta|$, all demand states remain consistent, and CDP requires searching the full state space, which is hard under Assumption~\ref{ass:hardness}. Increasing transparency (decreasing $\sigma^2$) monotonically reduces $|\mathcal{A}(\sigma^2)|$ and hence the difficulty of CDP, weakening Assumption~\ref{ass:hardness} and facilitating collusion.
\end{proof}

\begin{remark}
\label{rem:transparency}
This result challenges the conventional regulatory wisdom that transparency promotes competition. The European Commission's push for algorithmic transparency, the SEC's emphasis on trade reporting, and real-time pricing mandates may inadvertently facilitate the very collusion they seek to prevent. The mechanism is computational: transparency reduces the information-processing burden that currently prevents firms from monitoring each other's compliance with tacit agreements.
\end{remark}

%=============================================================================
\section{Extensions}
\label{sec:extensions}
%=============================================================================

\subsection{The AI transition}
\label{sec:ai-transition}

In practice, firms do not jump discontinuously from $\PP \neq \NP$ to $\PP = \NP$. AI systems expand computational capabilities gradually. I model this as firms having access to algorithms that solve problems of size up to $s$, where $s$ grows with AI capability (see Figure~\ref{fig:transition}).

\begin{definition}[Computational Capacity]
\label{def:capacity}
Firm $i$ has computational capacity $s_i \in \N$, meaning it can solve any instance of CDP (and CSP, OPP) of size at most $s_i$ in unit time. For instances of size $|\Gamma| > s_i$, the firm cannot solve the problem.
\end{definition}

\begin{proposition}[Regime Shifts]
\label{prop:phases}
As the minimum computational capacity $\underline{s} = \min_i s_i$ increases, the market passes through three regimes:
\begin{enumerate}
    \item[(i)] \textbf{Competitive regime} ($\underline{s} < s^*$): No firm can solve CDP for the actual market. Collusion is unsustainable. Prices converge to $p^*$.
    \item[(ii)] \textbf{Unstable regime} ($s^* \leq \underline{s} < s^{**}$): Some firms can partially solve CDP. Collusion is intermittent, with price wars triggered by detection failures. Market outcomes oscillate between $p^*$ and $p^M$.
    \item[(iii)] \textbf{Collusive regime} ($\underline{s} \geq s^{**}$): All firms can solve CDP for the actual market. Collusion is sustainable. Prices converge to $p^M$.
\end{enumerate}
The thresholds $s^*$ and $s^{**}$ depend on market complexity $|\Gamma|$.
\end{proposition}

\begin{proof}
When $\underline{s} < s^*$, no firm can detect deviations, and the argument of Direction~2 in Theorem~\ref{thm:main} applies. When $\underline{s} \geq s^{**}$, all firms can detect deviations, and Direction~1 applies. The intermediate regime follows from partial detection: firms can detect ``large'' deviations but not ``small'' ones, leading to a mixed equilibrium where firms occasionally test the boundaries of detection capability, triggering punishment phases.
\end{proof}

\begin{remark}
\label{rem:transition}
The AI transition is currently underway. Algorithmic pricing systems process millions of price points in real time, effectively solving CDP for simple markets (few products, stable demand). As AI systems grow more capable, the computational threshold $s^*$ is being crossed in progressively more complex markets. The empirical evidence of algorithmic collusion \citep{calvano2020,assad2024,fish2024} corresponds to early instances of this phase transition.
\end{remark}

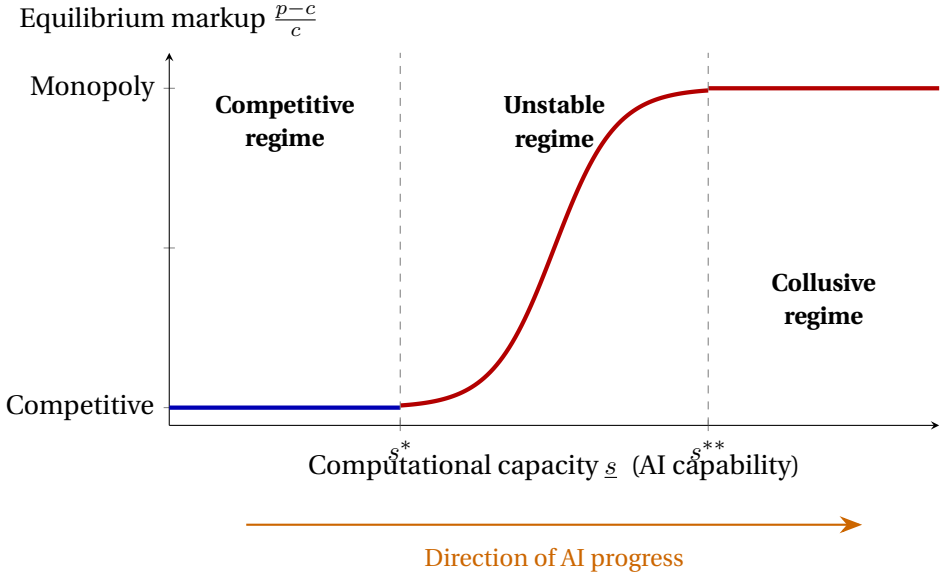
\begin{figure}[t]
\centering
\begin{tikzpicture}
\begin{axis}[
    width=0.85\textwidth,
    height=6.5cm,
    xlabel={Computational capacity $\underline{s}$\; (AI capability)},
    ylabel={Equilibrium markup $\frac{p - c}{c}$},
    xmin=0, xmax=10,
    ymin=0, ymax=1.05,
    xlabel style={at={(0.5,-0.05)}, anchor=north},
    xtick={3, 7},
    xticklabels={$s^*$, $s^{**}$},
    ytick={0.05, 0.5, 0.95},
    yticklabels={Competitive, , Monopoly},
    axis lines=left,
    every axis y label/.style={at={(ticklabel* cs:1.02)}, anchor=south},
    clip=false,
]
% Competitive regime
\addplot[blue!70!black, ultra thick, domain=0:3, samples=50] {0.05};
% Transition regime
\addplot[red!70!black, ultra thick, domain=3:7, samples=100] {0.05 + 0.9*(1/(1+exp(-2.5*(x-5))))};
% Collusive regime
\addplot[red!70!black, ultra thick, domain=7:10, samples=50] {0.95};
% Regime labels
\node[align=center, font=\small] at (axis cs:1.5, 0.85) {\textbf{Competitive}\\\textbf{regime}};
\node[align=center, font=\small] at (axis cs:5, 0.85) {\textbf{Unstable}\\\textbf{regime}};
\node[align=center, font=\small] at (axis cs:8.5, 0.35) {\textbf{Collusive}\\\textbf{regime}};
% Dashed vertical lines
\draw[dashed, gray] (axis cs:3, 0) -- (axis cs:3, 1.05);
\draw[dashed, gray] (axis cs:7, 0) -- (axis cs:7, 1.05);
% Arrow for AI progress
\draw[-{Stealth[length=3mm]}, thick, orange!80!black] (axis cs:1, -0.28) -- (axis cs:9, -0.28);
\node[font=\small, orange!80!black] at (axis cs:5, -0.38) {Direction of AI progress};
\end{axis}
\end{tikzpicture}
\caption{The AI transition. As firms' computational capacity increases, markets pass through three regimes. In the competitive regime ($\underline{s} < s^*$), firms cannot solve the collusion detection problem, and prices converge to marginal cost. In the unstable regime ($s^* \leq \underline{s} < s^{**}$), partial detection enables intermittent collusion. In the collusive regime ($\underline{s} \geq s^{**}$), full detection sustains monopoly pricing. AI advances push markets rightward along this curve.}
\label{fig:transition}
\end{figure}

\subsection{Heterogeneous computational capabilities}

When firms differ in computational capacity, asymmetric outcomes emerge.

\begin{proposition}[Asymmetric AI Adoption]
\label{prop:asymmetric}
Suppose firms $1, \ldots, n$ have computational capacity $s_i \geq s^{**}$ (``AI firms'') and firms $n+1, \ldots, N$ have capacity $s_j < s^*$ (``traditional firms''). Then:
\begin{enumerate}
    \item[(a)] AI firms can collude among themselves, sustaining supra-competitive prices in the products they sell.
    \item[(b)] Traditional firms cannot sustain collusion and price competitively.
    \item[(c)] The market equilibrium involves a two-tier pricing structure: AI firms earn monopoly rents while traditional firms earn competitive returns.
    \item[(d)] AI firms have a strategic incentive to increase market complexity (product differentiation, dynamic pricing) beyond the detection capacity of traditional firms, widening the computational gap.
\end{enumerate}
\end{proposition}

\begin{proof}
Parts (a)--(c) follow from Theorem~\ref{thm:main} applied separately to the AI-firm subgame and the traditional-firm subgame. AI firms with $s_i \geq s^{**}$ can solve CDP for the market among themselves, so Direction~1 of Theorem~\ref{thm:main} applies to their sub-coalition: collusion is sustainable. Traditional firms with $s_j < s^*$ cannot solve CDP, so Direction~2 applies: their collusion unravels. The two-tier structure follows from the resulting equilibrium: AI firms price at $p^M$ among themselves while traditional firms price at $p^*$. For part (d), note that AI firms can increase $|\Gamma|$ by introducing product variants or dynamic pricing schemes. If the resulting complexity exceeds $s^*$ for traditional firms but not $s^{**}$ for AI firms, the computational gap widens, reinforcing the two-tier structure.
\end{proof}

\subsection{Approximate collusion and bounded rationality}

Even if $\PP \neq \NP$, AI systems may approximately solve NP-hard problems using heuristics, achieving near-collusive outcomes.

\begin{definition}[$\alpha$-Approximate Collusion]
\label{def:approx}
An outcome is \textit{$\alpha$-approximately collusive} if joint profits satisfy $\sum_i \pi_i \geq \alpha \cdot \sum_i \pi_i(p^M)$ for $\alpha \in (0, 1]$.
\end{definition}

\begin{proposition}[Approximate Collusion]
\label{prop:approx}
Suppose firms have access to a polynomial-time \textit{probabilistic detection mechanism} $\mathcal{D}$ that, given observed data $(p^t, q^t)$, correctly classifies deviations from the collusive agreement with probability at least $\alpha \in (1/2, 1]$. Then for $\delta$ sufficiently close to $1$, firms can sustain $\alpha$-approximately collusive outcomes as a perfect public equilibrium.
\end{proposition}

\begin{proof}
Replace exact detection in the equilibrium construction of Theorem~\ref{thm:main}(a) with the probabilistic mechanism $\mathcal{D}$. A deviating firm is detected and punished with probability at least $\alpha$, and escapes detection with probability at most $1 - \alpha$. The incentive compatibility condition becomes
\[
    \pi_i^M \geq (1 - \alpha)\, \pi_i^D + \alpha\!\left[\pi_i^D + \frac{\delta}{1-\delta}\, \pi_i^P\right] \cdot \frac{1-\delta}{\delta},
\]
which is satisfied for $\delta$ sufficiently close to $1$ whenever $\alpha > 1/2$. The achievable collusive profit is bounded below by $\alpha \cdot \sum_i \pi_i(p^M)$ because the imperfect detection limits the severity of credible punishment. The equilibrium construction follows the folk theorem with imperfect public monitoring \citep{fudenberg1994}, where the signal structure is determined by $\mathcal{D}$.
\end{proof}

\begin{remark}
\label{rem:approx}
This extension is practically significant. The NP-hardness of exact CDP does not protect competition if approximate detection suffices, and modern AI systems are powerful approximate solvers. Large language models and deep reinforcement learning agents learn effective detection heuristics from market data, achieving $\alpha$ values well above $1/2$ in real markets. The competitive boundary is not a knife-edge; it is a gradient that AI is steadily climbing.
\end{remark}

%=============================================================================
\section{Empirical implications}
\label{sec:empirical}
%=============================================================================

The theory generates several testable predictions.

\begin{enumerate}
    \item[\textbf{P1.}] \textbf{AI adoption and markups.} Industries with higher AI adoption in pricing should exhibit higher markups and lower price dispersion, controlling for concentration and other standard determinants of competition. This is consistent with \citet{assad2024}, who find that algorithmic pricing in German gasoline markets increased margins by 9\%.

    \item[\textbf{P2.}] \textbf{Market complexity and competition.} Controlling for AI adoption, markets with greater complexity (more products, more volatile demand, less transparency) should be more competitive, because the CDP is harder to solve. This generates the novel prediction that \textit{complexity protects competition}: firms in complex markets are harder to monitor, making collusion unsustainable even with AI.

    \item[\textbf{P3.}] \textbf{Regime shifts.} As AI capabilities increase (measured by model size, inference speed, or data access), markets should exhibit rapid transitions from competitive to collusive pricing as computational thresholds are crossed. These transitions should be detectable as structural breaks in price-cost margin time series.

    \item[\textbf{P4.}] \textbf{Transparency and prices.} Mandated increases in market transparency (e.g.,\ real-time price reporting) should lead to \textit{higher} equilibrium prices, not lower, by facilitating AI-mediated collusion detection. This prediction runs counter to conventional regulatory expectations.

    \item[\textbf{P5.}] \textbf{Collusion without communication.} AI-mediated collusion should be observable even without any communication between firms: no shared algorithms, no data pooling, no third-party pricing services. The computational capability alone is sufficient. This distinguishes computational collusion from traditional conspiracy, which requires coordination.
\end{enumerate}

\begin{remark}
\label{rem:antitrust}
Prediction P5 poses a fundamental challenge for antitrust enforcement. If collusion emerges from computational capability rather than communication, then the legal standard requiring evidence of ``agreement'' or ``conspiracy'' is structurally inadequate. The theory suggests that antitrust law must shift from prosecuting collusive \textit{intent} to regulating collusive \textit{capacity}, a fundamental reorientation of competition policy.
\end{remark}

%=============================================================================
\section{Policy implications: computational antitrust}
\label{sec:policy}
%=============================================================================

If computational complexity is what maintains competition, then regulators have a new lever: the computational structure of markets themselves.

\subsection{Market complexity as a policy variable}

Traditional competition policy focuses on market structure (number of firms, barriers to entry) and conduct (price-fixing, bid-rigging). The present theory adds a third dimension: \textit{market complexity}, the computational difficulty of the collusion detection problem as a function of market design.

\begin{definition}[Market Complexity Index]
\label{def:mci}
The \textit{complexity index} of a market game $\Gamma$ is the computational complexity of its associated CDP, measured as the size of the smallest instance for which no polynomial-time algorithm achieves detection accuracy exceeding $1/2$.
\end{definition}

Regulators can increase market complexity through several mechanisms:

\begin{enumerate}
    \item[(i)] \textbf{Product differentiation incentives.} Encouraging product variety and customization increases the dimensionality of the pricing problem, making collusion detection harder.

    \item[(ii)] \textbf{Demand opacity.} Introducing randomness or opacity into demand information (e.g.,\ stochastic auctions, randomized procurement) increases the noise in observed market data, raising the computational cost of disentangling demand shocks from deviations.

    \item[(iii)] \textbf{Asynchronous pricing.} Requiring staggered pricing updates (rather than simultaneous price changes) breaks the temporal structure that AI systems exploit for real-time detection.

    \item[(iv)] \textbf{Algorithmic diversity mandates.} Requiring firms to use diverse, independently developed pricing algorithms reduces the implicit coordination that arises when similar algorithms converge to similar strategies.
\end{enumerate}

\subsection{The Regulatory Trilemma}

The Efficiency--Competition Impossibility (Corollary~\ref{cor:impossibility}) implies a trilemma for policymakers. They must choose at most two of three desirable market properties:

\begin{enumerate}
    \item[(a)] \textbf{Efficiency}: Prices reflect all available information.
    \item[(b)] \textbf{Competition}: Prices converge to marginal cost.
    \item[(c)] \textbf{AI integration}: Firms use computationally powerful AI systems.
\end{enumerate}

\noindent With AI integration, markets become efficient (prices quickly incorporate information) but collusive (AI enables detection and punishment). Without AI, markets can be competitive but informationally inefficient. Achieving both efficiency and competition requires limiting AI, at the cost of the productivity gains that AI provides.

This trilemma formalizes the emerging policy tension: governments simultaneously promote AI adoption for economic growth and worry about AI-mediated market manipulation. The theory shows that this tension is not a failure of policy design but a \textit{mathematical necessity}.

%=============================================================================
\section{Discussion}
\label{sec:discussion}
%=============================================================================

\subsection{Relationship to the original P = NP result}

\citet{maymin2011} established: \textit{Efficiency} $\iff$ $\PP = \NP$. The present paper establishes: \textit{Competition} $\iff$ $\PP \neq \NP$. Together:

\begin{table}[h]
\centering
\caption{The Efficiency--Competition Impossibility. Markets can be efficient (top-left) or competitive (bottom-right), but not both. AI pushes markets from the bottom-right cell toward the top-left cell.}
\label{tab:impossibility}
\begin{tabular}{@{}lcc@{}}
\hline
& $\PP = \NP$ & $\PP \neq \NP$ \\
\hline
Market efficiency & \checkmark & $\times$ \\
Market competition & $\times$ & \checkmark \\
\hline
\end{tabular}
\end{table}

Since $\PP \neq \NP$ is widely believed (and most of cryptography, including blockchain and digital signatures, relies on this assumption), the combined prediction is: \textit{markets are competitive but not efficient.} This is broadly consistent with the empirical evidence: markets exhibit persistent mispricing \citep{shiller2000,baker2007} alongside robust competition in most sectors.

The AI revolution is changing this equilibrium. By pushing the effective computational boundary toward $\PP = \NP$ for market-relevant problems, AI is shifting markets from the bottom-right cell (competitive, inefficient) toward the top-left cell (efficient, collusive). The question for society is whether the efficiency gains are worth the loss of competition.

\subsection{Why AI collusion is different from human collusion}

Traditional collusion is fragile because humans face cognitive limitations: bounded memory, emotional responses (anger at defectors, temptation to cheat), and inability to process complex multivariate data in real time. These limitations make the CDP effectively unsolvable for human managers, even in moderately complex markets.

AI agents face none of these limitations. They can:
\begin{itemize}
    \item Process the complete history of price-quantity data in milliseconds.
    \item Detect subtle patterns in competitor behavior across thousands of products.
    \item Compute optimal punishment strategies using reinforcement learning.
    \item Execute punishment with perfect commitment (no emotional forgiveness).
    \item Update strategies at a frequency that prevents human oversight.
\end{itemize}

\noindent The shift from human to AI decision-making is not a quantitative increase in collusion risk; it is a \textit{qualitative regime change}, a computational phase transition from the competitive to the collusive equilibrium.

\subsection{Limitations}

The ``only if'' direction (Theorem~\ref{thm:main}(b)) relies on Assumption~\ref{ass:hardness}, which requires that the specific market's CDP instance is hard, not merely that CDP is hard as a problem class. This is a genuine strengthening of $\PP \neq \NP$, and I have argued (Remark~\ref{rem:hardness}) that it holds generically. Markets with special algebraic structure in their demand functions (separability, low rank, sparsity) may violate the assumption even when $\PP \neq \NP$, making collusion feasible in those markets. This is consistent with empirical observation: collusion is easier in simple, transparent markets (few products, stable demand) than in complex ones.

The theorem also assumes rational, profit-maximizing firms. If some firms pursue objectives other than profit (e.g.,\ market share, social responsibility), or if bounded rationality constrains strategy choice beyond computational limitations, the competitive regime may persist longer than the theorem predicts.

A related concern is focal point coordination \citep{schelling1960}. Even without solving the CDP formally, firms might sustain partial collusion through simple focal strategies (e.g.,\ ``match the industry leader's price''). Such strategies do not require NP-hard computation but may sustain above-competitive prices in simple markets. The main theorem addresses this by focusing on \textit{sufficiently complex} markets where simple focal rules fail: with many products, volatile demand, and heterogeneous costs, there is no obvious focal price, and sustaining collusion requires the full detection-punishment apparatus formalized here.

Finally, the compact representation of the market game (Section~\ref{sec:model}) means that the state space $\Theta$ can be exponentially large relative to the description size, which is what makes the CDP hard. In markets where the relevant state space is small or easily enumerable, the NP-hardness results do not bind, and collusion may be feasible by brute-force search. The theorem's predictive force is for markets of increasing complexity, where the exponential state space makes detection intractable.

%=============================================================================
\section{Conclusion}
\label{sec:conclusion}
%=============================================================================

This paper proves that competitive market outcomes are sustained by computational limitations. If $\PP = \NP$, collusion is the equilibrium outcome; if $\PP \neq \NP$ and the market's collusion detection problem is hard on the instances that naturally arise, competition prevails. The result complements and mirrors \citet{maymin2011}: informational efficiency requires $\PP = \NP$, while competition requires $\PP \neq \NP$. Together, these results establish a fundamental impossibility: markets cannot be simultaneously efficient and competitive.

Artificial intelligence is pushing markets across this boundary. As AI systems expand firms' effective computational capabilities, the collusion detection problem that previously exceeded human capacity becomes solvable. The result is a mathematical prediction: sufficiently capable AI agents will sustain collusive outcomes as equilibria, without any explicit coordination or intent.

The implications for competition policy are immediate. Antitrust law built on detecting collusive \textit{intent} or \textit{communication} is structurally inadequate for a world in which collusion emerges from \textit{computation}. The paper proposes computational antitrust: the principle that market complexity is itself a competitive safeguard, and that regulators should design markets to be computationally hard to collude in, just as cryptographers design systems to be computationally hard to break.

Competition, the foundation of market economies since Adam Smith, is not a natural law. It is a consequence of our limitations. As those limitations recede, so may competition. Whether we can design institutions that preserve the benefits of competition in a world of unbounded computation remains to be seen.

%%%%%%%%%%%%%%%%%%%%%%%%%%%%%%%%%%%%%%%%%%%%%%
%% Bibliography                             %%
%%%%%%%%%%%%%%%%%%%%%%%%%%%%%%%%%%%%%%%%%%%%%%

\end{document}